\documentclass[conference]{IEEEtran}
\IEEEoverridecommandlockouts

\usepackage{cite}
\usepackage{amsmath,amssymb,amsfonts,amsthm}
\usepackage{algorithmic}
\usepackage{graphicx}
\usepackage{textcomp}
\usepackage{xcolor}
\usepackage[utf8]{inputenc}
\def\BibTeX{{\rm B\kern-.05em{\sc i\kern-.025em b}\kern-.08em
    T\kern-.1667em\lower.7ex\hbox{E}\kern-.125emX}}
\newcommand{\cambio}[1]{{\color{red!70!black}#1}}

\newcommand{\cP}{\mathbb{P}}
\newtheorem{theorem}{Theorem}

\renewcommand{\textfloatsep}{10pt plus 2pt minus 2pt}
\begin{document}

\title{Correlation Bounds and Markov Analysis for\\Ring-Oscillator TRNGs: A Joint Validation Framework}


\author{

\IEEEauthorblockN{Miguel Alcocer}
\IEEEauthorblockA{\textit{Informática y Estadística} \\
\textit{Universidad Rey Juan Carlos}\\
Móstoles, Spain \\
m.alcocer.2022@alumnos.urjc.es}
\and 
\IEEEauthorblockN{Ana Isabel Gómez}
\IEEEauthorblockA{\textit{Informática y Estadística} \\
\textit{Universidad Rey Juan Carlos}\\
Móstoles, Spain \\
ana.gomez.perez@urjc.es}

\and
\IEEEauthorblockN{Domingo Gómez-Pérez}
\IEEEauthorblockA{\textit{Estadística y Computación} \\
\textit{Universidad de Cantabria}\\
Santander, Spain \\
domingo.gomez@unican.es}
}

\maketitle

\begin{abstract}
True Random Number Generators (TRNGs) based on ring oscillators require rigorous statistical validation to ensure cryptographic quality. While the Mauduit-S\'ark\"ozy $k$-th order correlation measure $C_k$ provides theoretical bounds on pseudorandomness, and Maurer's Universal Statistical Test offers empirical entropy assessment, no prior work has correlated these metrics. This paper presents the first joint validation framework linking Maurer's Z-score to off-peak 2nd-order correlation $C_2$. We also derive the mathematical relationship between the previous two measures and high-order Markov chain transition probabilities in counter-based TRNGs over oscillator sampling architectures. Our results are validated computationally using OpenTRNG implementations, and demonstrate that practical implementations achieve 
Schmidt's improved bound.
The initial results suggest a strong positive correlation
between Maurer Z-score and $C_2$. Therefore, the results suggest a unified metric for TRNG quality-assessment can be achieve as a combination of these metrics, simplifying the study of new designs.
\end{abstract}


\section{Introduction}

Random bits are essential for cryptographic applications, providing unpredictable keys, nonces, and parameters of secure protocols. In practice, there are deterministic algorithms that, provided a \emph{true random sample}, called the \emph{seed}, generate arbitrarily long sequences of \emph{pseudorandom bits}.
True Random Number Generators (TRNGs) usually exploit a physical unpredictable process as a source of randomness to generate numbers with good statistical properties and high unpredictability. Depending on the application and available resources on hardware platforms,  various TRNG architectures have been proposed and  there are competitive candidates without a consensus on which is the winner. 
Free running oscillators are popular due to their
performance and simplicity. They are  based on connected logic stages, primarily based on inverter circuits, which negate the logic state of a signal. When certain conditions on the number of stages and topology are met, they
induce self-sustained oscillations at a given frequency. Even in this restricted class of TRNGs, multiple designs have been proposed, each
with a design tailored to cover specific needs (see for a review~\cite{almaraz2025current}). The designs  exploit thermal jitter as an entropy source, but their outputs require careful statistical validation.  The binary output sequences must be checked in order to verify their statistical properties. It is also important to implement online tests which provide early warnings against malfunctions or deliberate attacks.
As an example of this methodology, Intel Corp~\cite{jun1999intel} proposed a design that validates using Diehard, FIPS 140-1 and Knuth’s batteries, and also makes use of the suite for Intel RNGs that includes: Block Means Spectral analyses, Random walk test, Block Mean correlations, Block means, Periodogram, Spectral analyses, autocorrelations, 16-bit Maurer test, among others. On the other hand, the entropy requires a stochastic model to provide bounds per bit, among with a method to verify the assumptions of the estimated entropy via tests~\cite{killmann2011proposal}.

Ring oscillators (ROs) are  based on several inverters that are connected on a loop, as displayed on Figure \ref{fig:RO}.
In a Ring Oscillator-based TRNGs (RO TRNGs) a single event propagates through the stages to  generate an oscillating clock signal. This event travels around the ring, activating each inverter or buffer in sequence generating an oscillating
clock signal where the random behaviour is associated with the accumulated noise around sampling of the event (jitter noise).
\begin{figure}[!ht]
\centering
\includegraphics[scale=0.5]{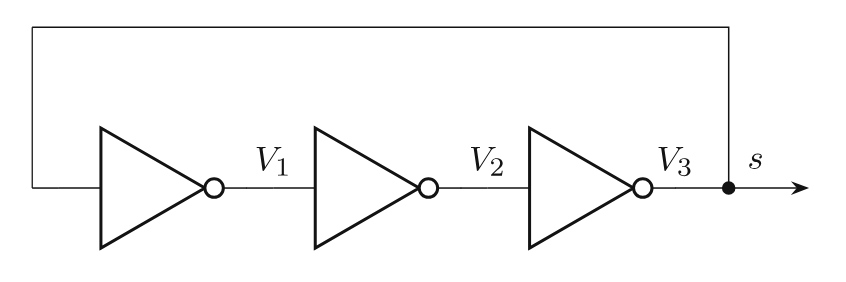}
\label{fig:RO}
\caption{Ring oscillator composed of three inverters $(V_1,V_2,V_3)$ generating an oscillating clock signal $s$.  Source~\cite{almaraz2025current}}
\end{figure}

Recently, Pebay-Peyroula et al. \cite{opentrng2024} released OpenTRNG, an open-source framework facilitating reproducible Ring Oscillator TRNGs design and experimental evaluation. Several designs include the following reference sampling architectures: Elementary Ring Oscillator (ERO), Multi Ring Oscillator (MURO) and Coherent Sampling Ring Oscillator (COSO). The designs are sorted by complexity, however they have in common that one or several ring oscillators are sampled by a clock signal and the output can be processed for better performance.

Baudet et al. \cite{baudet2011} pioneered the theoretical analysis of coherent-sampling RO-TRNGs, modelling output as a long-memory Markov process. Allini et al. \cite{allini2018} demonstrated that counting oscillator periods (rather than sampling) produces a higher-quality randomness, advocating Allan variance over standard variance for jitter characterization.
Their work showed empirically that Mauduit-Sárközy Correlation measure drops dramatically under the counter method. 
In related work, Saarinen \cite{Marku2021} showed  entropy and Mauduit-Sárközy second order correlation measure  and gave an algorithm to calculate block frequencies. We remark that entropy-per-bit is the main quality measure for the output of the TRNG. This motivated
Lubicz and Fischer \cite{lubicz2024} to provide algorithms to compute Shannon and min-entropy for RO-TRNGs using High-order Markov chains (8-bit memory, yielding $256\times256$ transition matrices). Their framework uses maximum-likelihood fitting and FFT-based convolution to aggregate entropy over multiple oscillators. 
Finally, Benea et al. \cite{benea2024flicker} shows that adding flicker (1/f) noise unexpectedly \emph{reduces} output autocorrelation in dual-oscillator sampling TRNGs
This suggests that practical implementations may achieve lower 2nd order correlation measure than predicted by thermal-jitter-only models.

Despite extensive research on RO-TRNGs \cite{allini2018, baudet2011, lubicz2024}, no prior work has empirically linked Mauduit-Sárközy measure, Maurer Test and High-order Markov Chains. In particular,  no prior work has computationally studied Maurer Z-score against 2nd order correlation measure as a joint validation or explored whether RO-TRNGs follow the ideal theoretical distribution for  each of these measures. Our work fills this gap by providing empirical and theoretical connections between these three measures.

The outline of the paper is as follows: Section~\ref{sec:background} provides the necessary background. Then, Section \ref{sec:Theoretic} provides theoretical bounds between the studied measures, that are validated by computer simulations in Section \ref{sec:Results}. Finally, Section \ref{sec:Conclusions} ends with conclusions and future work.

\section{Background}
\label{sec:background}
This section provides necessary background and recalls definitions for the convenience of the reader in order to make the paper self-contained.

The Mauduit-S\'ark\"ozy $k$-th order correlation \cite{mauduit1997} of a binary sequence $S = (s_1, s_2, \ldots, s_{N}) \in \{0,1\}^N$ is defined as:
\begin{equation}
    C_k(S) = \max_{D,M}\frac{1}{M}  \left| \sum_{n=1}^{M} (-1)^{s_n + s_{n+d_1} + \cdots + s_{n+d_{k-1}}} \right|,
\end{equation}
where the maximum is over all $1\le M\le N$ distinct lag vectors $D = (d_1, \ldots, d_{k-1})$ with $0 < d_1 < \cdots < d_{k-1} < N-M$, and indices are taken modulo $N$. For $k=2$, this coincides with the maximum Pearson correlation over all non-zero lags \cite{mauduit1997}.
The authors also define the so-called \emph{normality measure of order $k$}, first
define 
\begin{multline*}
T_k(S,M, V)\\=|\{n\;|\; 1\le n< M-k,\quad (s_{n+1},\ldots, s_{n+k})=V \}|.
\end{multline*}
The normality measure is defined as the maximum of 
$$
N_k(S) = \max_{V,M} \left|\frac{T^s_k(S,M,V)}{M} - \frac{1}{2^k}\right|.
$$
It is known that this is less than $C_k(S)$~\cite[Proposition 1]{mauduit1997}.

The correlation measure can be calculated by Fast Fourier transform in  $O(N^{k-1}\log N)$
In particular, the correlation of order 2 can be defined as follows:
\begin{equation*}
    C_2(S) = \max_{1 \leq \tau < N} \left| \text{ACF}(\tau) \right|,
\end{equation*}
where $\text{ACF}(\tau)$ is the standard normalized autocorrelation at lag $\tau$.

Maurer's test \cite{maurer1992} takes as input three integers  $k,r,L$. Take  $M=(r+L) k$ and a subsequence of $S$ of $M$ consecutive elements denoted by $S^M$.  This test measures the compressibility of $S^M$ by computing the average logarithmic spacing between repeated $k$-bit blocks. For this purpose, the test function $f_{T_U}$ is defined as follows,
\begin{equation}
\label{eqMaurerDefinition}
f_{T_U}(S^M)=\frac{1}{L} \sum_{n=r}^{r+L-1}\log_2 A_n(S^M)
\end{equation}
Let $b^k_n(S^M)= [s_{k\cdot(n-1)+1},\ldots,s_{k\cdot n}]$  be the $nth$ k-bit block of $S^M$, then
\begin{equation*}
 A_n(S^M)  = 
\begin{array}{@{}r@{}}

\begin{cases}
\min\{\, i \ge 1 : b^k_n(S^N) = b^k_{n-i}(S^N) \\ 
n,\ \  \text{if } \forall i < n,\; b^k_{n-i}(S^N) \neq b^k_n(S^N), 
\end{cases}
\end{array}
\end{equation*}
Then, a sample is rejected if the number of standard deviations exceeds a constant, following the $Z$ parameter that is defined as follows.  
\begin{equation}
    Z = \frac{f_{T_U} - \mathbb{E}[f_{T_U}]}{\sqrt{\text{Var}[f_{T_U}]}},
\end{equation}
where 
$\mathbb{E}[\cdot]$,
$\text{Var}[\cdot]$ are theoretical values for random sequences, usually approximated by heuristics estimates{} \cite{coron1998accurate}.

Finally, we introduce informally \emph{High order Markov chains} as seen in  \cite[equation 7]{lubicz2024}. A Markov chain of memory $k$ is a sequence $X=(X_i)$  of random variables taking values in $\{0,1\}$ and is defined by $T(X)$, the transition matrix. The transition matrix values are calculated by the conditional probabilities, in other words, 
\begin{equation}
T(X)= \frac{\cP(b_1^k(X)= (v_1,\ldots, v_k))}{\cP(b_{1}^{k-1}(X) = (v_1,\ldots, v_{k-1}))}.
\end{equation}
for a binary vector $V$ of length $k$.
Notice that for testing cryptographic-quality TRNGs, $T(X)$ values are required to be $1/2\pm k/\sqrt{N}$. For the Maurer Test the standard heuristic is $k=7$ and $Q=1280$, then  $|Z| < 2$ indicates the sequence is indistinguishable from random ~\cite{nist800-22}.

Finally for given a random sequence, it is expected that  the correlation measure satisfies the following bound,
\begin{equation*}
C_k(S) \leq 5\sqrt{\frac{k\cdot \text{ln}(N)}{N}},
\label{Eq:MaduitDelon}
\end{equation*}
with high probability\cite{alon2007}.

This has been improved by Schmidt \cite{schmidt2014},
\begin{equation}
C_k(S) \leq \cambio{\sqrt{\frac{2k\cdot \text{ln}(N)}{N}}}
\label{eq:C_bound}
\end{equation}
with \cambio{$\approx 3.5$} improvement \cambio{over the Alon constant~$5$ (factor $5/\sqrt{2}\approx 3.54$, independent of~$k$)}.






As closing remark, Saarinen \cite{Marku2021} derived analytical autocorrelation models for RO jitter, measuring $2\cdot \cP(s_i = s_{i+k}) - 1$  and providing closed-form entropy formulas for practical use.

\section{Theoretical Analysis}
\label{sec:Theoretic}
In this section, first we establish the relationship between $C_k(S)$ and Markov transition probabilities.
\begin{theorem}
For a  binary sequence $S$ of length $N$, suppose that $T(X)$ is the empirical transition matrix of a Markov chain of memory $k$. Under the assumption that $2^{k-2}\cdot C_k(S) < 1$, all transition probabilities satisfy,
\[
\left | \frac{\cP(b_n^k(X)= (V_1,\ldots, V_k))}{\cP(b_{n}^{k-1}(X) = (V_1,\ldots, V_{k-1}))} - \frac{1}{2}\right |
\le \frac{2^{k-2}C_k(S)}{1-2^{k-2}C_k(S)}. 
\]
for any binary vector $V$ of length $k$. 
\end{theorem}

\begin{proof}
Our aim is to prove the coefficients of the transition matrix in the Markov chain model of order $k$ are close to $1/2 + 2^{k-1}C_k(S)$ if $C_t(S)\ll 1,~ \forall t \le k$.
The proof is the following
\begin{multline*}
\frac{\cP(s_{n+1} = V_1, s_{n+2} = V_{2},\ldots, s_{n+k}= V_k)}{\cP( s_{n+1} =
V_{1},\ldots, s_{n+k-1}= V_{k-1})}
    =\\
    \frac{T_k(S,N, V)}{T_{k-1}(S,N, (V_1,\ldots, V_{k-1}))}\le \\
    \frac{1/2^k + N_k(S)}{1/2^{k-1} - N_{k-1}(S)}. 
\end{multline*}
By subtracting $1/2$ and using $2 N_{k-1}(S) \le N_k(S)$
\begin{multline*}
\frac{1/2^k + N_k(S)}{1/2^{k-1} - N_{k-1}(S)}- \frac{1}{2} = \\
\frac{ N_k(S)+1/2 N_{k-1}(S)}{1/2^{k-1} - N_{k-1}(S)}
\le \frac{2^{k-2}\cdot N_k(S)}{1- 2^{k-2}\cdot N_k(S)}.
\end{multline*}
Noticing that $N_k(S) \le C_k(S)$ by \cite[Proposition 1]{mauduit1997}, the result follows. 
\end{proof}

This implies that low $C_k(S)$ guarantees near-uniform $k$-bit pattern frequencies, a necessary condition for high entropy. For 8-bit Markov chains (as used in \cite{lubicz2024}), empirical transition probabilities $p_{ij} \approx 1/2 \pm 2^6 \cdot C_8(S)/(1-2^6 C_8(S)) \cdot C_8(S)$. The condition $2^6\cdot C_8(S) < 1$ is achievable for sequences of length of order \cambio{$10^6$} by Equation \eqref{eq:C_bound}.

Although it is difficult to compute $C_8(S)$ explicitly , we believe that the connection with Schmidt's bound helps to detect bias in the data. 

We now present our second theoretical result.
\begin{theorem} The Maurer's test is related to the correlation function of order $k$ by the following bounds\\
\begin{multline*}
\frac{(2^k - 2^{k+1}\cdot C_k(S))}{M\Big(\frac{1}{2^k} + C_k(S)\Big)} \cdot \log\!\Bigg(L \cdot \frac{(1 - 2 C_k(S))}{M\Big(\frac{1}{2^k} + C_k(S)\Big)}\Bigg)\\
\leq f_{T_U} 
\leq\\ 2^k\cdot\left (\log\left (L/M\right) - \log\left(\frac{1}{2^k}-C_k(S)\right)\right). 
\end{multline*}
\end{theorem}
\begin{proof}
Denote the number of blocks  equal to $V=(V_1,\ldots, V_k)$ by $T_k(S,M,V)$. The sum of the spacings between occurrences of $V$ is equal to the position of the last appearance of $V$ in $S^M$ minus the position of the first appearance.

We begin proving the upper bound, so we notice that the sum of spacings  is less than L. 
We remark that the logarithm function is a concave function in $(0,\infty)$ so, by Jensen's inequality, the maximum value of $f_{T_U}(S^M)$ is achieved when the elements are evenly spaced.
Therefore, substituting in Equation~\eqref{eqMaurerDefinition}
\begin{multline*}
f_{T_U}(S^M)=\frac{1}{L} \sum_{n=r}^{r+L-1}\log_2 A_n(s^M) \le\\ \frac{1}{L}\left(
\sum_{V}\sum_{n=r}^{r+L-1} \log\left(\frac{L}{T(S,M,V)}\right) 
\right) \le \\ \frac{1}{L}\left(
\sum_{V}\sum_{n=r}^{r+L-1}\log(L)-\log(M)- \log\left(\frac{1}{2^k}-N_k(S)\right)\right) = \\
2^k\cdot\left (\log\left (L/M\right) - \log\left(\frac{1}{2^k}-N_k(S)\right)\right), 
\end{multline*}
where the sum in $V$ is over all possible $(V_1,\ldots, V_k)$
For the other inequality, again we give a lower bound for the sum of spacings. Notice that the first appearance has to be in a position less than $r+L\cdot (N_k(S))$ and the last position has to be bigger than $r+L\cdot (1- N_k(S))$.
Therefore, the sum of spacings is greater than $L\cdot (1 - 2 N_k(S))$. By Markov's inequality, there are at least
$L\cdot (1 - 2 N_k(S))/T_k(S,M,V),$ spacings that are bigger than 
$L\cdot (1 - 2 N_k(S))/T_k(S,M,V).$ So, substituting
\begin{multline*}
   f_{T_U}(S^M) \\
   \ge  \frac{1}{L}\left( \sum_V L\cdot \frac{(1 - 2 N_k(S))}{T_k(S,M,V)}\log\left (L\cdot \frac{(1 - 2 N_k(S))}{T_k(S,M,V)}\right)
   \right) \ge\\ \frac{(2^k - 2^{k+1}\cdot N_k(S))}{M\Big(\frac{1}{2^k} + N_k(S)\Big)} \cdot \log\!\Bigg(L \cdot \frac{(1 - 2 N_k(S))}{M\Big(\frac{1}{2^k} + N_k(S)\Big)}\Bigg).
\end{multline*}
Using ~\cite[Proposition 1]{mauduit1997}, this finishes the proof.
\end{proof}

\section{Computational Results and Discussion}
\label{sec:Results}

Our design follows the differential counter-based approach of Allini et al. \cite{allini2018}. Two free-running ring oscillators (ROs) with slightly different frequencies $f_1$ and $f_2$ generate jittered clock signals. A reference clock at frequency $f_{\text{ref}} < \min(f_1, f_2)$ samples the oscillators synchronously. Instead of directly sampling bit values (which yields poor randomness \cite{baudet2011}),  oscillator periods are counted between reference edges using time-to-digital converters (TDCs). The XOR of least-significant bits from both counters forms the raw output:
\begin{equation}
    X_i = \text{LSB}(\text{Count}_1[i]) \oplus \text{LSB}(\text{Count}_2[i]).
\end{equation}
This differential extraction cancels global (common-mode) jitter, isolating independent thermal noise sources \cite{allini2018}.
It is usual to perform XOR accumulation of $n$ independent sources heuristically reducing correlation as $C_k(\bigoplus_{i=1}^n X_i) \approx C_k(X)^n$ \cite{allini2018} and \cite{gomez2022improved} for dependent case in STR-TRNGs.

We explore the Elementary Ring Oscillator (ERO) parameter space using OpenTRNG \cite{opentrng2024}, varying:
\begin{itemize}
    \item Oscillator frequencies: $f_1 \in [100, 250]$ MHz, $f_2/f_1 \in [0.7, 1.4]$
    \item Reference frequency: $f_{\text{ref}} \in [50, 150]$ MHz
    \item XOR accumulation depth: $n \in \{1, 2, 4, 8, 16\}$
\end{itemize}
For each configuration, we generate $N = 10^5$ bits and compute both $C_2$ and Maurer Z-score.

\subsubsection*{Ideal Noiseless Model}

To establish a baseline, we simulated an ideal noiseless system: two square waves at $f_1 = 192.5$ MHz and $f_2 = 136.5$ MHz sampled at $f_{\text{ref}} = 136.5$ MHz. The resulting bit sequence exhibits $C_2 (S) \approx 0.86$ (near-perfect correlation) and fails all randomness tests, confirming the noise requirements of the design.

\subsubsection*{Schmidt's Improved Bound}
Following Equation \eqref{eq:C_bound}, for $k=2$ and $N=10^5$, this yields \cambio{$C_2 (S) \lesssim 0.021$}, compared to the bound by Alon et al. \cite{alon2007} of 0.076. While derived for any random sequence, this provides a target for TRNG designs to achieve $C_2 (S) \ll 0.01$.




\subsection{Results and Discussion}

We measured $C_2(S)$ and Maurer Z-score across 58 parameter configurations (34 ERO and 24 COSO), using the output provided by ERO-TRNG. Additionally we compare to the output of a hardware baseline that implemented a RO-Oscillator based on counting. We summarize our findings in  Fig. \ref{fig:scatter} that shows the joint distribution of $C_2,Z$. 

Some key observations:
\begin{itemize}
    \item Strong positive correlation: Pearson $r = 0.95$ (95\% CI: $[0.91, 0.97]$).
    \item Configurations with $|Z| < 0.1$ consistently achieve $C_2 < 0.01$.
    \item High $C_2$ values ($>0.8$) correlate with $Z > 1$, indicating deviation from random expectation.
    \item XOR accumulation ($n \geq 4$) move the performance of the  configuration  closer to the ideal region ($|Z| \approx 0$, $C_2 < 0.01$).
\end{itemize}

\begin{figure*}[!ht]
    \centering
    \begin{minipage}[b]{0.49\textwidth} 
        \centering
        \includegraphics[height=8.5cm, width=\linewidth, keepaspectratio]{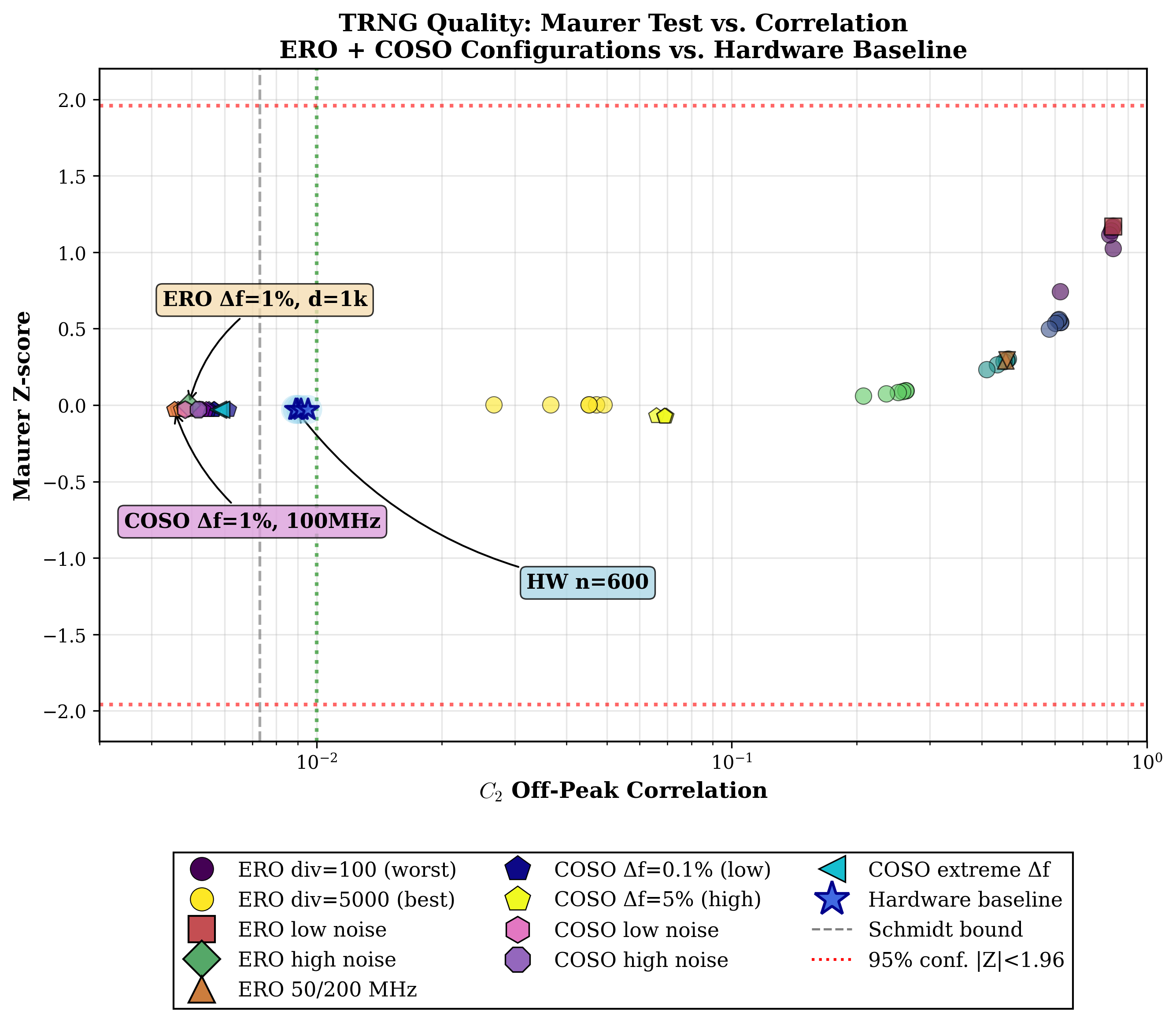}
        \caption{Maurer Z-score vs. off-peak $C_2(S)$ for 58 OpenTRNG configurations. Color indicates XOR accumulation depth ($n=1$ to 16). Dashed line: Schmidt bound in Equation \eqref{eq:C_bound}.}
        \label{fig:scatter}
    \end{minipage}
    \hfill 
\end{figure*}





\subsubsection*{Comparison with Sampling Method}

Following \cite{allini2018}, we compared counter-based extraction  with direct bit sampling. For identical oscillator parameters, the sampling method yielded $C_2 = 0.34$ and $Z = -2.1$ (failed test), while counting achieved $C_2 = 0.009$ and $|Z| < 1$ (passed). This $\approx 40\times$ improvement in $C_2$ agrees with
the improved statistical properties of counter-based designs.

\subsubsection*{Run Statistics}

Bit run lengths provide an intuitive quality metric. Our ideal model (noiseless square waves) exhibits mean run length $\mu = 1.22$, $\sigma = 0.41$. 
For example, hardware baseline shows consistent statistics for counter accumulations of 300--2400, with $\mu \approx 2.0$, $\sigma \approx 1.41$, matching the theoretical random expectation ($\mu = 2$, $\sigma = \sqrt{2} \approx 1.41$). This improvement directly correlates with reduced $C_2(S)$.

\section{Conclusion and future work}
\label{sec:Conclusions}
This paper makes three key contributions:
\begin{itemize}
    \item We establish the first empirical relationship between Maurer's Z-score and off-peak $C_2$ correlation, showing strong positive correlation across OpenTRNG parameter spaces.
    \item We derive explicit relations between $C_k$ and $k$-th order Markov chain transition probabilities, and also between Maurer's Test.

    \item  We provide initial computer simulations using OpenTRNG that provide an insight on the behaviour of Ring Oscillator TRNGs. We validated this framework , showing that counter-based differential extraction yields $C_2(S)$ values 2-3 times below Schmidt's improved theoretical bound.

\end{itemize}

Our work provides TRNG designers with a unified framework: low Maurer Z-scores ($|Z| < 2$) and low off-peak $C_2$ ($< 0.01$) jointly certify output quality, while deviations in either metric signal design flaws.
We presented the first systematic study linking Maurer's Universal Test to Mauduit-S\'ark\"ozy correlation bounds in ring-oscillator TRNGs.
For TRNG designers, our results suggest a practical validation protocol: compute both $C_2$ (via FFT autocorrelation) and Maurer Z-score (via standard implementations); configurations satisfying both $C_2 < 0.01$ and $|Z| < 2$ are likely to pass comprehensive test suites (NIST, AIS31). This unified metric reduces evaluation time and provides early detection of design flaws.
Our Markov chain analysis rigorously connects $C_k$ to transition probabilities, providing a bridge between combinatorial and probabilistic TRNG evaluation.

Future work will explore architecture-specific correlation bounds (tighter than the generic $O(\sqrt{k\ln N/N})$) and extend the framework to multi-source XOR accumulation and higher order correlations.

\section*{Acknowledgment}
The authors want to acknowledge the computational resources provided by the LABCOVI (Laboratorio de Computación y Visualización Avanzada) that belongs to the ``Centro de apoyo técnológico (CAT)'' of Universidad Rey Juan Carlos.

The work of the second author is partially supported by grant PID2023-151238OA-I00 funded by MICIU/AEI
/10.13039/501100011033 and by ERDF, EU.

The work of the third author 
is part of the
``CÁTEDRA UNIVERSIDAD DE CANTABRIA-INCIBE DE NUEVOS RETOS EN CIBERSEGURIDAD '', financed by “European Union NextGeneration-EU, the
Recovery Plan, Transformation and Resilience, through INCIBE.

\bibliographystyle{IEEEtran}
\bibliography{references}

\end{document}